\documentclass[journal,draftcls, onecolumn, 12pt]{IEEEtran}
\usepackage{amsfonts}
\usepackage{amssymb}
\usepackage{amsmath}
\usepackage{amsthm}
\usepackage{amsfonts}
\usepackage{color}
\usepackage{graphicx}
\usepackage{latexsym}
\usepackage{cite}
\usepackage{comment}
\usepackage{stfloats}
\usepackage{hyperref}

\usepackage{times}

\newcommand{\urlBiBTeX}[1]{\url{#1}}
\newcommand{\urlbibteX}[1]{\url{#1}}
\def\BibTeX{{\rm B\kern-.05em{\sc i\kern-.025em b}\kern-.08em
    T\kern-.1667em\lower.7ex\hbox{E}\kern-.125emX}}

\def\Ain{A}
\def\Aout{D}
\def\N{\mathcal{N}}
\def\E{\mathcal{E}}  %%%   for "expectation value" (in math mode)

\def\S{\mathcal {S}}
\def\C{\mathcal {C}}

\newtheorem{theorem}{Theorem}%[section]
\newtheorem{corollary}{Corollary}
\newtheorem{lemma}{Lemma}

\theoremstyle{definition}
\newtheorem{definition}{Definition}

\begin{document}

\title {A General Per-Flow Service Curve for GPS}

\author{
         Almut Burchard*, J\"{o}rg Liebeherr** \\[20pt] 
        * Department of Mathematics,   University of Toronto,  Canada. \\[-10pt]
        ** Department of ECE,   University of Toronto, Canada. \\[-5pt]
	 	E-mail:   almut@math.toronto.edu, jorg@ece.utoronto.ca.
        }%

\setcounter{page}{1}
\maketitle
\thispagestyle{plain}
\pagestyle{plain}

%\tableofcontents

\begin{abstract}
Generalized Processor Sharing (GPS), which provides the theoretical underpinnings 
for fair packet scheduling algorithms, has been studied extensively. However, 
a tight formulation of the available service of a flow only exists for traffic 
that is regulated by affine arrival envelopes and constant-rate links.
In this paper, we show that the universal service curve 
by Parekh and Gallager can be extended to concave arrival envelopes and links with time-variable 
capacity. We also dispense with the previously existing assumption of a stable system. 
\end{abstract}

%%%%%%%%
\section{Introduction}
\label{sec:intro}
Generalized Processor Sharing (GPS)  \cite{ParGal93,ParGal94} 
provides the foundation for fair packet scheduling algorithms, a class of traffic algorithms that seek to achieve a (weighted) max-min fair allocation of  the  
link bandwidth between individual or groups of traffic flows. 
GPS is an idealized algorithm in that it takes a fluid-flow view of traffic and allows a link 
to concurrently transmit traffic from arbitrarily many traffic flows. In contrast, a packet scheduler can only transmit one packet at a time and cannot interrupt the transmission of 
a packet. The relevance of GPS to fair packet scheduling algorithms is that the 
departure times of packets in some algorithms, e.g., Weighted Fair Queueing (WFQ) \cite{keshav-WFQ}, occur no later than the transmission time of a single packet of maximum size 
after the departure times with GPS. 

The service available to a flow in GPS is expressed in terms of a 
service curve, which is a function that expresses the amount of guaranteed departures 
of a flow in a time interval where the flow is backlogged.\footnote{In the terminology of 
the network calculus \cite{Book-LeBoudec}, such a service curve is referred to as `strict'. In this paper we exclusively encounter strict service curves.}
The (strict) service curve of a flow at a link with rate $C$ and GPS scheduling 
can be computed from 
the  so-called {\em universal service curve}  derived in \cite{ParGal93,ParGal94}, given by  
%Expressions for the available service at a link with Generalized Processor Sharing  
%scheduling derived in \cite{ParGal93,ParGal94} are widely credited for the first formulation 
%of a service curve. 
%In the min-plus linear system theory of the network calculus 
%\cite{Chang-minpluslinear,LeBoudec-minpluslinear}, service curves are system functions 
%of a network element that hold independently of the network load and the type of network traffic. 
\begin{align}
\S(t) =  {\max_{M \subseteq \N}} 
\frac{Ct - \sum_{j\in M} (\sigma_j + \rho_j\,t)}{\sum_{j\not\in M}\phi_j}\, ,  
\label{eq:universalservice}
\end{align}
where $\N$ is the set of flows and $(\phi_j)_{j\in\N}$ are so-called weights. The service curve 
makes the assumption  that (1)  the arrival traffic of each flow $j \in \N$ 
in a time interval of length $\tau$ is bounded by an affine arrival envelope $\E_j (\tau) = \sigma_j + \rho_j  \tau$, 
and (2) the system is stable in the sense that the total average arrival rate 
does not exceed the link capacity ($\sum_{j \in \N} \rho_j \le C$). 
The universal service curve yields a (strict) per-flow service curve $\S_i (t) = \phi_j \S(t)$ 
for a flow $j \in \N$. 
For general scenarios, where the $\E_j$ are not necessarily affine and 
the system may be unstable ($\sum_{j \in \N} \rho_j > C$),  
a pessimistic estimate for the available service 
can be given by the minimum guaranteed rate  $\tfrac{\phi}{\sum_{j \in \N} \phi_j} C$. 
This estimate can be somewhat improved 
by using knowledge of the arrival envelopes  
\cite{QiuKnightly99,QiuKnightly00,LiBuLi07}. 
If the  envelopes of the arrivals in~\eqref{eq:universalservice}  are replaced by 
 envelopes for departures, a generalization to non-affine envelopes is easily achieved. 
As pointed out in \cite[Sec. IV.C]{FidlerSurvey}, since a departure envelope of a flow can be expressed as a min-plus deconvolution of its arrival envelope and per-flow service curve, this only results in implicit expressions for (minimum) per-flow service curves. 

In this paper, we provide the following extensions to the per-flow (strict) service curves obtained from the universal service curve in~\eqref{eq:universalservice}:
\begin{itemize}
\item Arrival envelopes can be arbitrary concave functions; 
\item The link may have a time-variable capacity;
\item The link  need not be stable.
\end{itemize}
These relaxations are achieved by generalizing the concepts of {\em feasible ordering} in  \cite{ParGal93} and {\em feasible partition} in \cite{ZLiZhang98}. Note that the extension to time-variable service rates 
enables the computation of the available service for hierarchical 
schedulers~\cite{HPFQ}. 
We will show that the derived service curve is best-possible. 

In Sec.~\ref{sec:main} we state the main result. 
We provide a brief description of max-min 
fairness in Section~\ref{sec:feasible},
and then introduce the key notion of feasible subsets.
This notion is used in Sec.~\ref{sec:backlog} to derive backlog and output 
bounds. 
In  Section~\ref{sec:proof}, we prove
the main result, Theorem~\ref{thm:leftover}. 
Sec.~\ref{sec:wellposed} discusses GPS for general 
monotone arrival and service processes. 
We conclude the paper in Sec.~\ref{sec:concl}. 
  
\section{Statement of the main result}
\label{sec:main}

Let $\Ain$ and $\Aout$ denote the 
arrival and departure processes
for a flow or an aggregate of flows
arriving at a service element.
(Arrivals and departures for different flows will
be distinguished by subscripts).
The backlog is denoted by $B(t)=\Ain(t)-\Aout(t)$.
The cumulative service process of the element
will be described by a function $C(t)$.
The arrivals in a half-open interval $[s,t)$
are denoted by $\Ain(s,t):=\Ain(t)-\Ain(s)$,
and correspondingly for the departures and the service.
We always assume that arrival, departure, and
service processes are nondecreasing and left-continuous,
with $\Ain(t)=\Aout(t)=C(t)=0$ for $t\le 0$,
and $\Aout(t)\le \Ain(t)$ for all~$t$.

We say that the service element is {\em workconserving}, if
$\Aout(s,t)=C(s,t)$ on every interval that contains no idle period,
and $D(s,t)\le C(s,t)$ otherwise.
An important example is the constant-rate link, $C(t)=Rt$,
which serves traffic at 
the constant rate $R$ whenever the backlog is positive. 
In case the service process is given as
a time-varying rate $\dot C(t)$,
then the service element is workconserving if 
the departure rate satisfies $\dot \Aout(t)=\dot C(t)$
whenever there is a backlog at $t$.

Throughout this paper, we consider 
a finite set $\N$ of flows arriving to a service element.
Each flow $j\in\N$ is associated with a positive weight $\phi_j>0$.  

\begin{definition}
\label{def:WFQ}
A {\em Generalized Processor Sharing (GPS)} 
scheduler is a workconserving scheduling algorithm
which ensures that for any $0\le s<t$
and any flow $i\in\N$ 
that is backlogged on the entire interval $(s,t)$,
the departures satisfy
\begin{align}
\label{eq:WFQ}
\frac{\Aout_i(s,t)}{\phi_i} \ge 
\frac{\Aout_j(s,t)}{\phi_j}\quad \text{for all}\  j\in\N\,.
\end{align}
\end{definition}
Our main result provides a lower bound on $D_i(s,t)$
in terms of the parameters of the scheduler,
the service process, and the traffic arriving 
to each of the flows $j\in\N$.

To proceed, we need some more notation.
An {\em arrival envelope} for an arrival function $\Ain$ is
a nondecreasing function such that
\begin{align*}
\Ain (s,t) \le \E(t-s)
\quad \text{for all}\ 0\le s\le t\,. 
\end{align*}
We also say that the arrivals {\em comply} 
to $\E$ and write $\Ain \lesssim \E$.
By convention we set $\E(\tau)=0$ if $\tau \le 0$.
Without loss of generality, envelopes can be 
taken to be subadditive.

A nondecreasing function $\S$ is a {\em strict service curve}
for a flow at a service element if $\Aout(s,t) \ge \S (t-s)$
whenever the flow is backlogged on the entire 
interval $(s,t)$.
By convention, $\S(\tau)=0$ for $\tau\le 0$.
Without loss of generality, a strict service curve
may be taken to be superadditive and
nonnegative. 

We use $\C$ to denote the strict service curve offered by 
a workconserving service element. In general, 
$\C$ is a strict service curve~if
\begin{align*}
C(s,t)\ge \C(t-s)\quad \text{for all}\  0\le s\le t\,.
\end{align*}
In that case, we say that the service process
{\em complies} to $\C$ and write $C\gtrsim\C$.
As a special case, 
$\C(t)=Rt$ is the strict service curve
for the workconserving link with constant rate $R$.

\begin{theorem}[Leftover service curve]\label{thm:leftover}
Let $\N$ be a finite set of flows arriving
to a GPS scheduler, as in Definition~\ref{def:WFQ}.
Assume that $C\gtrsim\C$.
Fix $i\in\N$. For each $j\in \N\setminus\{i\}$, 
let $\E_j$ be an envelope with $\Ain_j \lesssim \E_j$. 
If $\C$ is convex and each $\E_j$ is concave in $t$,
then
\begin{align}
\label{eq:leftover}
\S_i(t) :=  \max_{M \subseteq \N\setminus \{i\}} 
\frac{\phi_i}{\sum_{j\not\in M}\phi_j} 
\Bigl(\C(t) - \sum_{j\in M}\E_j(t)\Bigr)
\end{align}
is the best-possible strict service curve for flow $i$.
\end{theorem}

We refer to $\S_i$ as the {\em leftover service curve} 
available to flow~$i$ under GPS. Note that there
are no hypotheses on the arrivals from flow~$i$.
If no envelope is available for some flow $j\in\N$,
a conservative estimate can be obtained by setting $\E_j(t)=+\infty$
for all $t>0$.

By construction, $\S_i$ is nonnegative, nondecreasing,
and convex in $t$, with $\S_i(0)=0$ and 
$\frac{\phi_i}{\sum_{j\in\N}\phi_j} \C\ \le\  \S_i\ \le\  \C$.
We will show 
that $\S_i(t)$ equals the service that flow $i$ 
receives in a scenario where it is backlogged on $(0,t)$,
the flows $j\ne i$ are greedy ($\Ain_j=\E_j$), 
and the service element is lazy ($C= \C$),
see Lemma~\ref{lem:greedy}.

Eq.~\eqref{eq:leftover} 
and the definition of the GPS scheduler
are reminiscent of expressions for
max-min fairness. In the proof of the theorem,
we will exploit this connection. 
The convexity and concavity assumptions
will play an important role.

%%%%%%%%
\section{Max-min fairness and feasible subsets}
\label{sec:feasible}

Let $\N$ be a collection of players. As in Section~\ref{sec:main}, let 
$(\phi_j)_{j\in \N}$  be positive weights.
Each player $j\in\N$ requests a nonnegative share
$x_j$ of a resource~$X$.
An allocation $(y_j)_{j\in\N}$ 
with $0\le y_j\le x_j$ for $j\in\N$
is {\em max-min fair},
if $\sum_{j\in\N} y_j = \min\bigl\{\sum_{j\in \N} x_j, X\bigr\}$,
and for each $i\in\N$ with $y_i<x_i$
\begin{align}
\label{eq:maxmin}
\frac{y_i}{\phi_i}\ge \frac{y_j}{\phi_j}\quad \text{for all}\ j\in \N\,.
\end{align}
Here, $y_i$ represents the share allocated to player~$i$.
The first condition requires
the allocation to be {\em waste-free}, 
that is, the entire resource must be used
unless the requests of all players are satisfied. 
Eq.~\eqref{eq:maxmin}
specifies that small requests are satisfied
in full while large requests are served
in proportion to their weights ($\phi_j$). It is known
that these conditions uniquely determine
the allocation.  Explicitly, 
$y_i=\min\{x_i,\phi_if\}$ with
\begin{align}
\label{eq:def-f}
f:= \max_{M\subset \N} \frac{X-\sum_{j\in M} x_j}
{ \sum_{j\not\in M} \phi_j}\,.
\end{align}

The value $f$ is called the {\em fair share}
associated with the allocation problem.
By convention, for $M=\N$ the fraction
takes the value $-\infty$ if the
numerator is negative and $+\infty$
otherwise.  The maximum is attained by the set of satisfied players,
\begin{align}
\label{eq:Msat}
M_{\rm sat}
:=\bigl \{j\in\N\ \big\vert\ x_j\le \phi_j f\bigr\}\,.
\end{align} 
Clearly, the fair share
is nonnegative and jointly convex in $x_j$ and $X$.
It is nondecreasing in $X$ and nonincreasing
in each $x_j$.  Its value
is finite if and only if $\sum_{i\in\N} x_i > X$,
and it satisfies the  lower bound
$f\ge \frac{X}{\sum_{j\in \N} \phi_j}$.

\smallskip
Different from Eq.~\eqref{eq:leftover},
the maximum in Eq.~\eqref{eq:def-f} ranges 
over {\em all} subsets $M\subset\N$.
The two formulas are related as follows.

\begin{lemma}
\label{lem:frac}
Let $M\subset \N$ be a non-empty subset, and $i\in M$.
Then either 
\begin{align*}
\frac{x_i}{\phi_i} 
\ \le\  \frac{X-\sum_{j\in M\setminus\{i\}}x_j}
{\sum_{j\not\in M\setminus\{i\}}\phi_j} 
\ \le \ \frac{X-\sum_{j\in M}x_j}
{\sum_{j\not\in M}\phi_j}\,
\end{align*}
or both inequalities are reversed.
\end{lemma}
\begin{proof} If $M=\N$, then
the inequalities hold if and only if $\sum_{j\in \N} x_j\le X$.
Otherwise, set $x'_i:=X-\sum_{j\in M} x_j$
and $\phi'_i:=\sum_{j\not\in M}\phi_j>0$, and write
\begin{align*}
\frac{X -\sum_{j\in M\setminus\{i\}}x_j}
{\sum_{j\not\in M\setminus\{i\}} \phi_j}
&= \frac{x_i+x'_i}{\phi_i+\phi'_i}\\
&= \lambda \frac{x_i}{\phi_i}  + (1-\lambda)\frac{x'_i}{\phi'_i}\\
&= \lambda \frac{x_i}{\phi_i}  + (1-\lambda)
\frac{X-\sum_{j\in M}x_j}{\sum_{j\not \in M}\phi_j}\,,
\end{align*}
where $\lambda=\frac{\phi_i}{\phi_i +\phi'_i}$ 
lies strictly between $0$ and $1$.
Therefore either both inequalities hold,
or both fail.
\end{proof}

As a consequence of the lemma,
the fair allocation to flow~$i$ can also
be computed by $y_i=\min\{x_i,f_i\}$, where
\begin{align}
\label{eq:def-fi}
f_i:= \max_{M\subset\N\setminus\{i\}} 
\frac{\phi_i} { \sum_{j\not\in M}\phi_j}
\Bigl(X-\sum_{j\in M}x_j\Bigr)\,.
\end{align}

We next consider the impact that a subset
of requests can have on a max-min fair allocation.

\begin{definition}
\label{def:feasible}
Let $M\subset\N$, and $X>0$.
A collection of requests $(x_j)_{j\in M}$ is 
{\em feasible}, if
\begin{align}
\label{eq:feasible}
\max_{j\in M}\frac{x_j}{\phi_j} \le 
\frac{X-\sum_{j\in M}x_j}{\sum_{j\not\in M} \phi_j}\,.
\end{align}
In that case, $M$ is called a {\em feasible subset}
of $\N$ for the data $(\phi_j)_{j\in\N}$, $(x_j)_{j\in M}$, and $X$.
\end{definition}

Feasibility of $(x_j)_{j\in M}$ 
means that $M_{\rm sat}$, 
the set of satisfied players from Eq.~\eqref{eq:Msat},
contains $M$,
regardless of the values in the set $(x_j)_{j\not\in M}$. 
Conversely, for any set 
of requests $(x_j)_{j\in\N}$,
the corresponding
subset $M_{\rm sat}$ is feasible.
By way of examples, 
a single request $x_i$ is feasible 
if $x_i\le \frac{\phi_i}{\sum_{j\in\N}\phi_j}X$.
A full set of requests $(x_j)_{j\in\N}$
is feasible if $\sum _{j\in \N} x_j\le X$.

\smallskip\noindent{\bf Remark.}
Feasible subsets are closely related
to the  notion of feasible orderings 
introduced in \cite[Sec. V.C]{ParGal93}.
By definition, a {\em feasible ordering} (``$\prec$'') is a total
order on $\N$ with the property that
\begin{align*}
\frac{x_k}{\phi_k}
<\frac{X-\sum_{j\prec k}x_j}{\sum_{j \succeq k}\phi_j}\quad
\text{for all}\ k\in \N\,.
\end{align*}
One can verify that for any feasible ordering, 
the downsets  $M_k:=\{j\mid j\preceq k\}$
are feasible subsets.  
Feasible subsets are also downsets
for the partial order
induced by the feasible partition 
constructed in \cite{ZLiZhang98}.

\smallskip
The next lemma will be used to construct
chains of feasible subsets.
In the case where $M=\N$ and $\sum_{j\in\N}x_j< X$, it implies that
orderings of $\N$ along which
the fraction $\frac{x_j}{\phi_j}$ 
is nondecreasing are feasible. 
This recovers Lemma 5 in \cite{ParGal93}. 
We note in passing that there exist other feasible 
orderings where $\frac{x_j}{\phi_j}$ is not monotone.

\begin{lemma} 
\label{lem:chain}
Let $(x_j)_{j\in  M}$ be a feasible subset for a resource $X>0$.
If $k\in M$ and $x_k$ satisfies
\begin{align*}
\frac{x_k}{\phi_k}=\max_{j\in M}\frac{x_j}{\phi_j}\,,
\end{align*} 
then $M\setminus\{k\}$ is feasible.
\end{lemma}

\begin{proof} By the maximality of $k$,
\begin{align*}
\max_{j\in M\setminus \{k\}}
\frac{x_j}{\phi_j}
\le \frac{x_k}{\phi_k}\le 
\frac{C-\sum_{j\in M\setminus\{k\}} x_j}
{\phi_k + \sum_{j\not \in M} \phi_j}\,,
\end{align*}
where the second inequality follows from
Eq.~\eqref{eq:feasible}
by Lemma~\ref{lem:frac}.
Thus $M\setminus\{k\}$
is feasible.
\end{proof}

Let $(y_j)_{j\in\N}$ be the max-min fair allocation of a resource
$X$ resulting from requests $(x_j)_{j\in\N}$.
Denote by $\bar y_i := x_i - y_i$ the {\em unmet demand} of player
$i$.  In terms of the fair share
from Eq.~\eqref{eq:def-f}, 
the unmet demand is given by
$\bar y_i=[x_i-\phi_i f]_+$.
Here, we have used the notation $[x]_+=\max\{x,0\}$.
The waste-free property of the allocation is
equivalent to 
$\sum_{j\in \N} \bar y_j = \bigl[\sum_{j\in\N} x_j - X\bigr]_+\!$.
The unmet demand satisfies the following useful inequalities.

\begin{lemma}
\label{lem:mono-contraction}
Let $(\bar y_j)_{j\in\N}$ be the unmet demands
in the max-min fair allocation of a resource $X$
resulting from requests $(x_j)_{j\in\N}$,
and let $(\bar y'_j)_{j\in\N}$
be defined accordingly from $X'$ and
$(x'_j)_{j\in\N}$.
Then 
\begin{align}
\label{eq:contraction}
\sum_{j\in \N} |\bar y_j -\bar y_j'|
\le \sum_{j\in \N} |x_j-x_j'|  + |X-X'|\,.
\end{align}
Moreover, we have the monotonicity property
\begin{align*}
\left.
\begin{array}{l}
x_j\le x_j'\ \text{for all} \ j\in \N\\
X\ge X'
\end{array}\right\}
\ \Longrightarrow\ \bar y_j \le \bar y'_j\ \text{for all}\ j\in\N\,.
\end{align*}
\end{lemma}
\begin{proof} We start with the second claim.
Fix $i\in\N$.  By definition, $\bar y_i=[x_i-f_i]_+$,
and correspondingly for $\bar y'_i$.
It is apparent from Eq.~\eqref{eq:def-fi}
that $f_i$ is nondecreasing
in $X$ and nonincreasing in the variables $x_j$ for $j\ne i$.
This proves monotonicity.

For Eq.~\eqref{eq:contraction}, let $(x_j)_{j\in\N}$
and $(x'_j)_{j\in \N}$ be as in the statement of the theorem.
Denote by
$(\bar z_j)_{j\in\N}$  the unmet demand resulting
from the requests $\min(\{x_j,x'_j\})_{j\in\N}$ for the resource
$\max\{X,X'\}$, and by $(\bar w_j)_{j\in\N}$ be the
unmet demand resulting from requests 
$\max(\{x_j,x'_j\})_{j\in\N}$ for the resource $\min\{X,X'\}$.
By monotonicity,
\begin{align*}
\bar z_j \le \bar y_j \le \bar w_j  \quad 
\text{for all } \  j\in\N \, ,
\end{align*}
and likewise for $\bar y'_j$.
Therefore 
\begin{align*}
\sum_{j\in \N} |\bar y_j -\bar y_j'|
&\le \sum_{j\in \N} (\bar w_j-\bar z_j)\\
&= \Bigl[\sum_{j\in \N} \max\{x_j,x_j'\}-\min\{X, X'\}\Bigr]_+\\
&\qquad - \Bigl[\sum_{j\in \N} \min\{x_j,x_j'\}-\max\{X,X'\}\Bigr]_+\\
&\le \sum_{j\in \N} |x_j-x'_j| + |X-X'|\,,
\end{align*}
where the second step used the waste-free property.
\end{proof}

The lemma implies that the max-min fair allocation
for a fixed value of $X$, viewed as a mapping
$(x_j)_{j\in\N}\mapsto (\bar y_j)_{j\in\N}$, 
contracts the $\ell^1$-distance and preserves the natural order.

%%%%%%%% 
\section{Performance bounds}
\label{sec:backlog}

The following theorem says that the aggregate cumulative 
departures from a feasible subset $(x_j)_{j\in  M}$
are at least as large as if each flow $j\in M$ were allocated
a dedicated link with service process $x_jC$.
Note than no assumption is made on busy periods.

\begin{theorem} [Departures]
\label{thm:output}
Let $(\Ain_{j}(t))_{j\in\N}$ be arrivals from
a set of flows to a GPS scheduler 
with service process $C(t)$.  Fix $M\subset\N$, 
and let $(x_j)_{j\in M}$ be a feasible subset of 
requests for the resource $X=1$. Then for all $t\ge 0$,
\begin{align}
\label{eq:output}
\sum_{j\in M} \Aout_j(t) \ge 
\sum_{j\in M} \inf_{s\le t} \bigl\{
\Ain_j(s) + x_j C(s,t)\bigr\}\,.
\end{align}
\end{theorem}
\begin{proof}
We proceed by induction on the number of elements of $M$.
When $M=\emptyset$, there is nothing to show.

For the inductive step, let
$M\subset \N$ be a non-empty feasible subset,
and suppose the claim has already been established for 
its proper feasible subsets.  
Choose $k\in M$ to maximize
the ratio $\frac{x_j}{\phi_j}$.
By Lemma~\ref{lem:chain}, $M\setminus\{k\}$ is feasible. 
By the inductive hypothesis,  for all $t\ge 0$,
\begin{align}
\label{eq:inductive}
\sum_{j\in M\setminus \{k\}} D_j(t)
\ge \sum_{j\in M\setminus\{k\}} 
\inf_{r\le t} \bigl\{A_j(r) + x_j C(r,t)\bigr\}\,.
\end{align}
Fix $t>0$ and let $s$ be the start of the busy period for flow $k$
that contains $t$. 
If $D_k(s,t)\ge x_k C(s,t)$, then 
\begin{align*}
D_k(t) = A_k(s) + x_k C(s,t)\,,
\end{align*}
since $D_k(s)=A_k(s)$.
Eq.~\eqref{eq:output} follows by adding
Eq.~\eqref{eq:inductive}. 

Otherwise, since flow $k$ is backlogged on $(s,t)$,
\begin{align*}
\frac{D_j(s,t)}{\phi_j} \le \frac{D_k(s,t)}{\phi_k}
 < \frac{x_k}{\phi_k}C(s,t)
\end{align*} 
for all $j\in \N$
by Eq.~\eqref{eq:WFQ}. Therefore
\begin{align*}
\sum_{j\not\in M} D_j(s,t) 
&< \Bigl(\sum_{j\not\in M} \phi_j \Bigr) \frac{x_k}{\phi_k}C(s,t)\\
&\le \Bigl(1-\sum_{j\in M}x_j\Bigr) C(s,t)
\,,
\end{align*}
where the second inequality is by the feasibility of
$(x_j)_{j\in M}$.
Since the scheduler is workconserving,
it follows that
\begin{align*}
\sum_{j\in M} D_j(s,t) 
&> \sum_{j\in M}x_j C(s,t)\,,
\end{align*}
and therefore
\begin{align*}
\sum_{j\in M} D_j(t)
&> \sum_{j\in M}\bigl(D_j(s) + x_j C(s,t)\bigr)\,.
\end{align*}
Clearly, $D_k(s)=A_k(s)$ by the choice of $s$.
For the flows $j\ne k$, we
use Eq.~\eqref{eq:inductive} at time $s$ to obtain
\begin{align*}
\sum_{j\in M\setminus\{k\}} \!\!\!\bigl(\!D_j(s)+x_jC(s,t)\!\bigr)
&\ge\!\!\!\!\sum_{j\in M\setminus\{k\}}
\!\!\!
\inf_{r\le s} \bigl\{\!A_j(r) + x_j C(r,t)\!\bigr\}\,.
\end{align*}
Eq.~\eqref{eq:output} follows once we add
the term for $j=k$ and
extend the range of the infima to $r\le t$.
This completes the induction.
\end{proof}

In the case where $M=\{i\}$, Theorem~\ref{thm:output}
yields 
\begin{align*}
D_i(t)  \ge \inf_{s\le t} 
\Bigl\{ \Ain_i(s) + \frac{\phi_i}{\sum_{j\in\N}\phi_j} C(s,t)\Bigr\}\,.
\end{align*}
More generally, the theorem implies the following key estimates.

\begin{corollary} [Backlog]
\label{coro:backlog}
Define
\begin{align*}
B_j^*(t) := \sup_{r\le t} 
\left\{A_j(r,t)-x_j C(r,t)\right\}
\end{align*}
for $j\in  M$. Under the assumptions of Theorem~\ref{thm:output},
\begin{align}
\label{eq:backlog}
\sum_{j\in M} B_j(t) \le \sum_{j\in M} B_j^*(t)\,,
\qquad t\ge 0\,.
\end{align}
\end{corollary}
\begin{proof} Write $B_j(t)=A_j(t)-D_j(t)$ and 
apply Eq.~\eqref{eq:output}.
\end{proof}

\begin{corollary} [Output burstiness]
\label{coro:upper}
Under the assumptions of Theorem~\ref{thm:output},
\begin{align*}
\sum_{j\in M} 
D_j(s,t)\le \sum_{j\in M} \bigl(B_j^*(t) +x_j C(s,t)\bigr)\,,
\quad 0\le s\le t\,.
\end{align*}
\end{corollary}
\begin{proof} 
By Theorem~\ref{thm:output},
\begin{align*}
\sum_{j\in M^*} D_j(s,t)
&\le \sum_{j\in M^*} \bigl( A_j(t)-D_j(s)\bigr)\\
&\le \sum_{j\in M^*}
\sup_{r\le s} \bigl\{ A_j(r,t)-x_j C(r,s)\bigr\}\\
&\le \sum_{j\in M^*} \bigl(B_j^*(t) +x_j C(s,t)\bigr)\,.
\end{align*}
In the last step, we have extended the range of
the supremum to $r\le t$ and applied the definition of $B_j^*(t)$.
\end{proof}

For later use, we note that if
$\Ain_j(t)\lesssim \sigma_j+\rho_j t$ and
$C(t)\gtrsim R(t-L)$ with $\rho_j\le x_jR$, then 
\begin{align}
\label{eq:B*}
B_j^*(t) \le \sigma_j +\rho_j L\, , \qquad t\ge 0\,.
\end{align}
Corollary~\ref{coro:backlog} implies 
Theorem 4 in \cite{ParGal93} as follows. 
The assumption in  \cite{ParGal93} is that
the arrivals comply to token-bucket envelopes,
$A_j\lesssim \sigma_j+\rho_j t$, 
that the link offers a constant-rate service
$C \gtrsim Rt$, and  that the
stability condition $\sum_{j\in\N}\rho_j <R$ holds.
If we choose $x_j=\frac{\rho_j}R$, then
$\sigma_j-\sigma_j^t$  
equals $B_j^*(t)-B_j(t)$, where 
$\sigma_j^t$ is defined in \cite{ParGal93} as the sum 
of the filling level of the token bucket and the backlog at time $t$.  
Further, in \cite{ParGal93} the set 
$M$ is assumed to be a downset for a feasible ordering of $\N$.
Under these assumptions, Eq.~\eqref{eq:backlog} 
reduces to the central conclusion in  \cite{ParGal93} that
$\sum_{j\in M}\sigma_j^t\le \sum_{j\in M}\sigma_j $.

%%%%%%%% 
\section{The leftover service curve}
\label{sec:proof}
%\color{blue}{ {\bf Eliminated lemma}}
Consider the definition of the leftover service curve
$\S_i$ in Eq.~\eqref{eq:leftover}. 
It follows from Lemma~\ref{lem:frac} that
\begin{align*}
\min\bigl\{\E_i(t),\S_i(t)\bigl\} = \min\bigl\{\E_i(t),\phi_i \S(t)\bigr\}\,,
\end{align*}
where
\begin{align}
\label{eq:S-uni}
\S(t) := \max_{M\subset \N} \frac{\C(t)-\sum_{j\in M}\E_j(t)}
{\sum_{j\not\in M}\phi_j}\,.
\end{align}
Note the structural similarities of Eq.~\eqref{eq:leftover} to Eq.~\eqref{eq:def-fi}, 
and of Eq.~\eqref{eq:S-uni} to Eq.~\eqref{eq:def-f}. 
In the special case where the envelopes $\E_j$ are affine,
$\S$ agrees with the universal service 
curve in Eq.~\eqref{eq:universalservice}.
The maximum in Eq.~\eqref{eq:S-uni} is attained by 
\begin{align}
\label{eq:M*}
M^*:= \{j\in\N\  \mid \E_j(t)\le \S_j(t)\}\,,
\end{align}
see Eq.~\eqref{eq:Msat}.

\begin{lemma} \label{lem:slopes}
Let $\N$, $\C$ and $\E_j$ be as in Theorem~\ref{thm:leftover}.
Given $\tau>0$, define $M^*$ by Eq.~\eqref{eq:M*}
with $t=\tau$ and $\E_i=+\infty$. Then 
\begin{align*}
x_j:=\frac{\dot \E_j(\tau_-)}{\dot \C(\tau_-)}\,,\quad
j\in M^*
\end{align*}
defines a feasible subset for the resource $X=1$.  
\end{lemma}
Here, we used the notation $f(x_-) = \sup_{y<x} f (y)$. 
\begin{proof} By Eqs.~\eqref{eq:S-uni} and~\eqref{eq:M*},
the subset of requests
$x_j':=\frac{\E_j(\tau)}{\C(\tau)}$, $j\in M^*$
is feasible for $X=1$. Since $\E_j(\tau)\ge \tau \dot \E_j(\tau_-)$
by concavity and $\C(\tau)\le \tau\dot \C(\tau_-)$ by
convexity, we have $x'_j\ge x_j$
for all $j\in M^*$. Thus,   
$(x_j)_{j\in M^*}$ is a feasible subset.
\end{proof}

We next consider the special case of token-bucket envelopes and
latency-rate service curves. (The general proof 
follows immediately afterwards.)

\begin{lemma} \label{lem:LB-LR}
Under the hypotheses of Theorem~\ref{thm:leftover},
suppose additionally that the service curve
has the form $\C(t)=R(t-L)$, and the envelopes
are given by $\E_j(t)=\sigma_j+\rho_j t$
for $j\in\N\setminus\{i\}$.
Then Eq.~\eqref{eq:leftover} defines a strict
service curve for flow $i$.
\end{lemma}
\begin{proof}
Suppose that flow $i$ is backlogged
on some interval $(s,t)$.
We need to show that $D_i(s,t)\ge S_i(t-s)$.

Set $\tau=t-s$. Let $M^*$ be as in 
Eq.~\eqref{eq:M*} with $\tau$ in place of $t$ and $\E_i=+\infty$, 
and set $x_j=\frac{\rho_j}{R}$
for $j\in M^*$. By Lemma~\ref{lem:slopes}, the 
subset of requests $(x_j)_{j\in M^*}$ 
is feasible for $X=1$.
By Corollary~\ref{coro:upper},
\begin{align*}
\sum_{j\in M^*} 
D_j(s,t)\le \sum_{j\in M^*} \left\{B_j^*(t) +x_j C(s,t)\right\}\,.
\end{align*}
Since the scheduler is workconserving, it follows that
\begin{align*}
\sum_{j\not\in M^*} D_j(s,t) 
&\ge \Bigl(\!1-\!\sum_{j\in M^*} x_j\!\Bigr)
C(s,t) -\!\sum_{j\in M^*} B_j^*(t)\\
&\hspace{-1cm}\ge \Bigl(\!R-\!\sum_{j\in M^*} \rho_j\!\Bigr)
(t\!-\!s\!-\!L)
-\!\sum_{j\in M^*} \left\{\sigma_j + \rho_j L\right\}\\
&\hspace{-1cm}=\C(t-s)-\sum_{j\in M^*} \E_j(t-s)\,.
\end{align*}
In the first line, the coefficient of $C(s,t)$
is nonnegative by the feasibility of $(x_j)_{j\in M^*}$.
In the second line, we have used that
$C(t)\gtrsim R(t-L)$ and applied Eq.~\eqref{eq:B*}.
In the last line, we have canceled the terms
$\rho_jL$
and inserted the envelopes and service curves.
By Eq.~\eqref{eq:WFQ},
\begin{align*}
D_i(s,t) &\ge \frac{\phi_i}{\sum_{j\not\in M^*} \phi_j} 
\sum_{j\not\in M^*} D_j(s,t)\\
& \ge \frac{\phi_i}{\sum_{j\not\in M^*} 
\phi_j} \Bigl(\C(t-s)-\sum_{j\in M^*}\E_j(t-s)\Bigr)\\
&=\S_i(t-s)\,.
\end{align*}
The final step used the maximality of $M^*$
in Eq.~\eqref{eq:leftover}.
\end{proof}

We are ready to tackle the main result.

\begin{proof}[Proof of Theorem~\ref{thm:leftover}]
Given $0\le s< t$, set $\tau=t-s$, and fix $i\in \N$.
For $j\in \N\setminus\{i\}$, 
consider the tangent line to the graph of $\E_j$
at $\tau$, defined by 
$\E'_j(u)= \sigma_j+\rho_j u$ with
\begin{align*}
\rho_j:=\dot \E_j(\tau_-)\,,\qquad
\sigma_j:= \E_j(\tau)-\rho_j\tau\ge 0\,.
\end{align*}
Since $\E_j\le \E'_j$ by concavity, 
the arrival process $A_j$ complies to the
token-bucket envelope $\E'_j$. 
Also consider the tangent line
to $\C$ at $\tau$, defined by $\C'(u)=R(u-L)$ with
\begin{align*}
R:=\dot\C(\tau_-)\,,\quad L:= \tau - \tfrac{\C(\tau)}{R}\ge 0\,.
\end{align*}
Since $\C\ge \C'$ by convexity, the
service process $C$ complies to the latency-rate service curve $\C'$.
By Lemma~\ref{lem:LB-LR},
\begin{align*}
\S'_i:= \max_{M \subseteq \N\setminus \{i\}} 
\frac{\phi_i}{\sum_{j\not\in M}\phi_j} 
\Bigl(\C' - \sum_{j\in M}\E'_j\Bigr)
\end{align*}
is a strict service curve for flow $i$. In particular,
if flow $i$ is backlogged on $(s,t)$ then
\begin{align*}
D_i(s,t)\ge  \S'_i(t-s) = \S_i(t-s)\,,
\end{align*}
where the equality is by the choice of $\tau=t-s$.
We conclude that $\S_i$ is a strict service curve.
By Lemma~\ref{lem:greedy} below, there are
scenarios where the departures saturate the service curve.
Therefore $\S_i$ is best possible.
\end{proof}

%%%%%%%% 
%\section{The greedy scenario}
%\label{sec:greedy}

%We now present a scenario where the the departures saturate
%the service curve in Eq.~\eqref{eq:leftover}.

\begin{lemma}
[The greedy/lazy scenario]
\label{lem:greedy} 
In the setup of Theorem~\ref{thm:leftover},
let the service process be
$C(t)=\C(t)$, and the arrival processes 
$\Ain_j(t)= \E_j(t)$ for 
$j\in\N$ and $t\ge 0$. 
Then
\begin{align}
\label{eq:greedy}
\Aout_j(t)=\min\bigl \{\E_j(t),\S_j(t)\bigr\}\,,\quad j\in \N\,.
\end{align}
\end{lemma}
\begin{proof} Let $t>0$ be given.
Since the scheduler is workconserving,
the aggregate departures satisfy
\begin{align*}
\sum_{j\in \N} \Aout_j(t) =
\inf_{0\le s\le t} \Bigl\{ \sum_{j\in \N} A_j(s) + C(s,t)\Bigr\}\,.
\end{align*}
Inserting the assumptions on the arrival and service
processes, we obtain
\begin{align}
\notag
\sum_{j\in \N} \Aout_j(t)
&= \inf_{0\le s\le t} \Bigl\{ \sum_{j\in \N} \E_j(s) + \C(t)-\C(s)\Bigr\}\\
\notag &= \min \Bigl\{ \sum_{j\in \N} \E_j(t), \C(t)\Bigr\}\\
\label{eq:greedy-proof}
&=\sum_{j\in \N} \min \bigl\{\E_j(t), \S_j(t)\bigr\}\,.
\end{align}
The second step follows
since the minimum is attained at $s=0$ or $s=t$ by concavity.
In the last step, we have used
that $y_j =\min\{\E_j(t),\S_j(t)\}$
is a max-min fair allocation of the resource
$X=\C(t)$, and therefore waste-free.

On the other hand, since $\S_j$ is a service curve
for flow~$j$,
\begin{align*}
D_j(t) &\ge \inf_{0\le s\le t} \bigl\{ \E_j(s) + \S_j(t)-\S_j(s)\bigr\}\\
&= \min\bigl\{\E_j(t),\S_j(t)\bigr\}\,.
\end{align*}
Since this holds for every $j\in\N$,
by Eq.~\eqref{eq:greedy-proof} it holds with equality.
\end{proof}

Lemma~\ref{lem:greedy} demonstrates that the departures from
a GPS scheduler in the greedy scenario
necessarily satisfy Eq.~\eqref{eq:greedy}.
For completeness of the argument, we 
show that these departures
actually conform to Definition~\ref{def:WFQ}. 
The workconserving property follows 
from the waste-free property of
the max-min fair allocation.
It remains to verify Eq.~\eqref{eq:WFQ} on
an arbitrary interval where flow $i$ is backlogged.

Eq.~\eqref{eq:greedy} yields 
$B_j(t)=\bigl[\E_j(t)-\S_j(t)]_+$.
By concavity, the ratio $\frac{B_j(t)}{t}$
is nonincreasing in $t$. Therefore, if
flow $i$ is backlogged at time $t$,
then it is backlogged for all $0<s\le t$.
By Eq.~\eqref{eq:maxmin},
$\frac{D_i(t)}{\phi_i} \ge \frac{D_j(t)}{\phi_j}$,
with equality if flow $j$ is backlogged as well.
If flow $j$ is backlogged at time $s$, then
$\frac{D_i(s)}{\phi_i} =\frac{D_j(s)}{\phi_j}$,
and Eq.~\eqref{eq:WFQ} follows.
Otherwise, flow $j$ is not backlogged at time
$s$, and $D_j(s)=\E_j(s)$. The difference
\[
\frac{D_i(s,t)}{\phi_i}- \frac{D_j(s,t)}{\phi_j}
= \frac{\S_i(t)-\S_i(s)}{\phi_i} - \frac{\E_j(t)-\E_j(s)}{\phi_j}
\] 
is concave in~$s$, and nonnegative at $s=0, t$.
Therefore it is nonnegative for every $0\le s\le t$,
proving Eq.~\eqref{eq:WFQ} also in this case.

%%%%%%%% 
\section{The backlog process}
\label{sec:wellposed}

We briefly address the question how to
describe the departures from a GPS scheduler 
with a general nondecreasing service process $C(t)$ 
and nondecreasing arrival 
processes $A_j(t)$, $j\in\N$. We will
argue that the workconserving property
together with Eq.~\eqref{eq:WFQ}
completely determines the backlog process, 
and hence the departures.

Consider once more the relation
between the GPS scheduler and max-min fairness,
as evidenced by Eq.~\eqref{eq:WFQ} and Eq.~\eqref{eq:maxmin}.
The departures $D_j(s,t)$ over a time interval $[s,t)$ 
define an allocation of the resource $X=C(s,t)$
among a set of flows $j\in\N$,
each of which requests a share $x_j=B_j(s)+A_j(s,t)$.
The backlog $B_j(t)$ plays the role of the 
unmet demand. 

On any interval where the arrival processes
$A_j(t)$ are concave and
$C(t)$ is convex, the departures are given by
the max-min fair allocation
\begin{align*}
D_j(s,t) = \min\bigl\{B_j(s)+A_j(s,t),\phi_j f\bigr\}\,,\quad j\in \N\,,
\end{align*}
where $f$ is defined
by Eq.~\eqref{eq:def-f} with $x_j=B_j(s)+A_j(s,t)$
and $X=C(s,t)$.
This follows by applying Lemma~\ref{lem:greedy}
to the
time-shifted processes $A'_j(\tau)=B_j(s) +A_j(s,s+\tau)$
and $C'(\tau)=C(s,s+\tau)$,
and then setting $\tau=t-s$.
The backlog satisfies the difference equation
\begin{align}
\label{eq:backlog-process}
B_j(t) = \bigl[B_j(s)+A_j(s,t)-\phi_j f\bigr]_+\,,\quad j\in\N\,.
\end{align}

However, Eq.~\eqref{eq:backlog-process}
cannot hold for general arrival and service processes
on arbitrary intervals.
Flows that are backlogged at time $t$  but are
idle at an earlier time $s<t$ receive less service than 
indicated by Eq.~\eqref{eq:backlog-process}.
The underlying reason is that Eq.~\eqref{eq:WFQ}
provides no explicit service guarantees for such flows.

Since Eq.~\eqref{eq:backlog-process} is valid when $s$ is so close to $t$
that the set of backlogged flows remains
constant from $s$ to $t$, taking the limit
$s\to t$ yields the differential equation 
\begin{align}
\label{eq:ODE}
\dot B_i(t) = 
\dot A_i(t)- 
\frac{\phi_i}{\sum_{j\not\in M(t)}\phi_j}
\Bigl(\dot C(t)\!-\!\!\sum_{j\in M(t)}\dot A_j(t)\Bigr)\,,
\end{align}
so long as  $B_i(t)>0$. 
Here, $M(t)=\{j\in\N\mid  B_j(t)=0\}$ is the set of 
flows that are not backlogged at time $t$.
The differential equation
holds at every time $t$ where
the arrival and service processes are differentiable,
except at instants where $M(t)$ changes.
%Apart from those instants,
%$M(t)$ maximizes the right hand side among its subsets.
%Thus the departure rate $\dot D_i(t)$ equals the
%max-min fair allocation to flow~$i$ of
%the resource $X'=\dot C(t)$ that results from the
% requests $x'_j=\dot A_j(t)$ for $j\in M(t)$ and 
%$x'_j=+\infty$ otherwise. 
(If the arrival and service processes
are not absolutely continuous, 
the differential equation should be supplemented 
by equations that account for their jumps 
and singular continuous components.)

Eq.~\eqref{eq:ODE} determines
the backlog process on intervals where $M(t)$ is constant.
These intervals in turn depend on the departures, rendering 
the differential equation nonlinear.
Standard theorems that guarantee the existence and 
uniqueness of solutions for nonlinear
differential equations do not apply, because the right 
hand side of Eq.~\eqref{eq:ODE} 
does not have the requisite continuity properties.  

We construct the backlog process as follows.
Given arrival and service processes
$A_j(t)$ and $C(t)$, we approximate them with
piecewise linear nondecreasing functions.
Specifically, we consider the class of functions that
are linear on intervals $(t_\ell,t_{\ell+1}]$,
where the breakpoints $t_\ell$ form an increasing
sequence with $t_0=0$ and $\lim t_\ell=+\infty$.
Jumps are permitted at each $t_\ell$. 
Since linear functions are 
simultaneously convex and concave,
Lemma~\ref{lem:greedy} implies that
the backlog process for the approximating
scenario satisfies Eq.~\eqref{eq:backlog-process}
on each interval $(t_\ell, t_{\ell+1}]$.
Then $B_j(t)$ and $\Aout_j(t)$
lie again in the piecewise linear class, 
with at most $|\N|$
additional breakpoints appearing
between $t_\ell$ and $t_{\ell+1}$ at instants where 
some flow ceases to be backlogged.
By Lemma~\ref{lem:mono-contraction}, all errors can be 
bounded explicitly in terms of the original discretization
error.  Consequently, the backlog process does not depend on
the precise approximation scheme that was used
in its construction. 

Thanks to Lemma~\ref{lem:mono-contraction},
the backlog evolves by an order-preserving 
family of contractions.
% \begin{theorem} [Properties of the backlog process] 
%Consider two GPS schedulers
%with service processes $C(t)$ and $C'(t)$.
%Let $\Ain_j(t)$ and $\Ain'_j(t)$
%be arrival processes for flows at the respective 
%schedulers, in both cases indexed by $j\in\N$. 
%Fix $s\ge 0$. If $C(s,t)=C'(s,t)$ and
%$\Ain_j(s,t)=\Ain'_j(s,t)$ for all $j\in \N$ and $t\ge s$,
%then
%\begin{align*}
%\sum_{j\in \N} |B_j(t)-B'_j(t)|\le \sum_{j\in \N} |B_j(s)-B'_j(s)|
%\end{align*}
%for all $t\ge s$.
%If, moreover, $B_j(s)\le B_j'(s)$ for all $j\in \N$, then
%the same inequalities hold at the later time $t$.
%\end{theorem}
One implication is that the backlog process at
a GPS scheduler with random stationary arrival
and service processes that is started with empty queues
is stochastically increasing, in analogy to \cite[Lemma~9.1.4]{Book-Chang}. As $t\to\infty$, 
the flows separate into two groups,
one consisting of underloaded  flows whose backlog 
process approaches a steady state, and the other of overloaded    
flows whose backlog becomes unbounded.

%%%%%%%%%%%%%% 
\section{Conclusions}
\label{sec:concl}
We have addressed a longstanding open problem in the theory of 
fair queueing algorithms, and extended the strict service 
curve formulation for GPS schedulers by Parekh and Gallager  to concave arrival 
envelopes and links with time-variable capacity. We show that 
the service curves holds under any load condition, and is not 
limited to stable systems. 
With this paper, the leftover service curve formulation for GPS 
has a comparable degree of generality as existing 
leftover formulations of other `classical' 
scheduling algorithms, such as Static Priority, 
FIFO, and Earliest-Deadline-First.

\section*{Acknowledgements}
This work is supported in part by the Natural Sciences and Engineering Research Council of Canada (NSERC).

\bibliographystyle{plain}

\end{document}